\documentclass[10pt,conference]{IEEEtran}
\usepackage[centertags]{amsmath}
\usepackage{amsfonts}
\usepackage{amssymb}
\usepackage{epsfig}

\usepackage{graphicx}
\usepackage[centertags]{amsmath}
 \DeclareMathOperator{\rank}{rank}

\renewcommand{\baselinestretch}{0.98}
\begin{document}
\title{On the Relation  Between the Index Coding and the Network Coding Problems }

\author{
\authorblockN{Salim El Rouayheb, Alex Sprintson, and  Costas Georghiades}
\authorblockA{Department of Electrical and Computer Engineering \\
Texas A\&M University, College Station, TX 77843 \\
\{salim, spalex, c-georghiades\}@ece.tamu.edu}}



%

\date{Draft of \today}
\maketitle
\newtheorem{theorem}{Theorem}
\newtheorem{lemma}[theorem]{Lemma}
\newtheorem{claim}[theorem]{Claim}
\newtheorem{proposition}[theorem]{Proposition}
\newtheorem{condition}[theorem]{Condition}
\newtheorem{observation}[theorem]{Observation}
\newtheorem{conjecture}[theorem]{Conjecture}
\newtheorem{property}[theorem]{Property}
\newtheorem{assertion}[theorem]{Assertion}
\newtheorem{example}[theorem]{Example}
\newtheorem{definition}[theorem]{Definition}
\newtheorem{corollary}[theorem]{Corollary}
\newtheorem{remark}[theorem]{Remark}
\newtheorem{note}[theorem]{Note}
\newtheorem{problem}[theorem]{Problem}

\begin{abstract}
In this paper we show that the Index Coding problem captures several
important properties of the more general Network Coding problem. An
instance of the Index Coding problem includes a server that holds a
set of information messages $X=\{x_1,\dots,x_k\}$ and a set of
receivers $R$. Each receiver has some side information, known to the
server, represented by a subset of $X$ and demands another subset of
$X$. The server uses a noiseless communication channel to broadcast
encodings of messages in $X$ to satisfy the receivers' demands. The
goal of the server is to find an encoding scheme that requires the
minimum number of transmissions.

We show that any instance of the  Network Coding problem can be efficiently
reduced to an instance of the Index Coding problem. Our reduction shows that
several important properties of the Network Coding problem carry over to the
Index Coding problem. In particular, we prove that both scalar linear and
vector linear codes are insufficient for achieving the minimal number of
transmissions.
%
\end{abstract}

\section{Introduction}

Since its introduction by the seminal paper of Ahlswede et al.
\cite{ACLY00}, the network coding paradigm has received a
significant interest from the research community (see e.g.,
\cite{FS07, YLC06} and references therein). Network coding extends
the functionality of the intermediate network nodes from merely
copying and forwarding their received messages to combining the information content of several incoming messages and forwarding the result over the outgoing edges. The network coding approach was shown to produce substantial gain over the traditional approach of routing and tree packing in many scenarios.

%

The Index Coding problem has been recently introduced in \cite{BBJK06} and has
been the subject of several studies \cite{WPCPC061,RCS07,LS07}. 
An instance of the Index Coding problem includes a
server/transmitter that holds a set of information messages $X$ and
a set of receivers $R$, each one of them has some side information
represented by a subset of $X$, known to the server, and demands
another subset of $X$. The server can broadcast  encodings of
messages in $X$ over a noiseless channel. The objective is to
identify an encoding scheme that satisfies the demands of all clients with the minimum number of transmissions.


Figure \ref{fig:butterfly} depicts an instance of  the Index Coding problem
that includes a  server with four messages $x_1,\dots, x_4\in\{0,1\}$ and four clients.
For each client, we show the set of messages it has (side information), and the set of messages it wants (demands). Note that the server can always satisfy the demands of the clients by
sending all the messages. However, this solution is suboptimal since it is
sufficient for the server to broadcast the two messages $x_1+x_2+x_3$ and
$x_1+x_4$ (all operations are over $GF(2)$).
This shows that by using an efficient encoding scheme, the server can significantly reduce the number of transmissions, and, in turn, reduce the
delay and the energy consumption.


In general, each message can be divided into several \emph{packets} and the
encoding scheme can combine packets from different messages to minimize the
number of transmissions. With \emph{linear} index coding, all packets are
considered to be elements of a certain finite field $\mathbb{F}$ and each transmitted packet is a linear combination of the packets  corresponding to the original messages in $X$.
The linear solutions can be further classified into \emph{scalar linear} and
\emph{vector linear}. With a scalar linear solution, each message corresponds
to exactly one packet, while with a vector linear solution each message can be
divided into several packets. Note that the example shown in Figure
\ref{fig:butterfly} uses a scalar linear solution over $\mathbb{F}=GF(2)$.

The Index Coding problem was studied from an information theoretical
perspective in \cite{WPCPC061}. The authors of \cite{BBJK06} established lower
and upper bounds on the minimum number of transmissions based on the properties
of a certain related graph. References \cite{RCS07} and \cite{CS08} present several heuristic solutions for this problem. In addition, the authors of
\cite{LS07} showed the suboptimality of scalar linear encoding schemes, which
disproves the conjecture of \cite{BBJK06}.

\begin{figure}[t]%
\begin{center}%
\epsfig{file=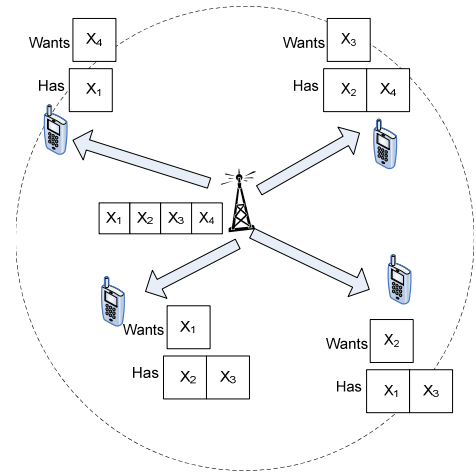, scale=0.3}
\end{center}
 \caption{An instance of the Index Coding problem.}
\label{fig:butterfly}
\end{figure}

\subsection*{Contributions}

Index coding can be seen as a special case of the Network Coding
problem. In this paper, we show that, nevertheless, several important
properties of the more general Network Coding problem carry over to the Index
Coding problem. To that end, we present a reduction that maps any instance of
the Network Coding problem to a corresponding instance of the Index Coding
problem. We use this reduction to establish several fundamental properties of
the Index Coding problem.

First, we show that a scalar linear solution may require more
transmissions than a vector linear one. In particular, we show two
instances of the Index Coding problem  in which a vector linear
solution that divides each message into two packets yields a smaller
number of transmissions than a scalar linear solution.

Second, we show that even vector linear solutions for the Index Coding problem
are insufficient for achieving the minimal number of transmissions.  In
particular, we use our reduction and the construction presented in \cite{DFZ05}
to show an instance of the Index Coding problem for which a non-linear solution
requires a lower number of transmissions than the linear one.

\section{Model}
\subsection{Network Coding}
Let $G(V,E)$ be a graph with vertex set  $V$  and edge set $E$. For each edge
$e(u,v)\in E$, we define the in-degree of $e$ to be the in-degree of its tail
node $u$. Similarly, we define the out-degree of $e$ to be the out-degree of
its head node $v$. Let $S\subset E$ be the subset of edges in $E$ of zero
in-degree and let $D\subset E$ be the subset of edges in $E$ of zero out-degree.
We refer to edges in $S$ and $D$ as \emph{input} and \emph{output} edges,
respectively. We denote $m=|E|$, $k=|S|$, $d=|D|$, and assume that the edges in
$E$ are indexed such that  $S =\{e_1,\dots,e_k\}$ and
$D=\{e_{m-d+1},\dots,e_m\}$. Also, for each edge  \mbox{$e=(u,v)\in E$}, we
define $\mathcal{P}(e)$ to be the set of the parent edges of $e$, i.e.,
\mbox{$\mathcal{P}(e)=\{(w,u);\ (w,u)\in E)\}$}.

We represent a communication network by a 3-tuple $\mathbb{N}(G(V,E),X,\delta)$
defined by an acyclic graph $G(V,E)$, a message set $X=\{x_1,\dots,x_k\}$, and
an onto function \mbox{$\delta: D\longrightarrow X$} from the set of output
edges to the set of messages. Each message $x_i\in X$ consists of a vector of $n$
packets $x_i=(x_{i1},\dots,x_{in})$.

We assume that the message $x_i$ is  available at the tail node of
the input edge $e_i$. The function $\delta$, referred to as
the \emph{demand} function, represents, for each output edge $e_i\in
D$, the message demanded by its head node.


\begin{definition}[Network Code]
Let $\mathbb{N}(G(V,E),X,\delta)$ be an instance of the Network Coding problem
with $k=|X|$ messages, each message is a vector of $n$ packets,
$x_i=(x_{i1},\dots,x_{in}) \in \Sigma^n$, where $\Sigma=\{0,\dots,q-1\}$ is a
$q$-ary alphabet. Then, an $(n,q)$ \emph{network code} of block length $n$ is a
collection $$C=\{f_e=(f_e^1,\dots,f_e^n); f_e^i:(\Sigma^n)^k\longrightarrow
\Sigma, e\in E, 1\leq i\leq n\}$$
of  \emph{global encoding functions}, indexed by the edges of $G$, that satisfy the following conditions:
 \begin{enumerate}
   \item [(N1)] $f_{e_i}(X)=x_i$ for $i=1,\dots,k$;
   \item [(N2)]$f_{e_i}(X)=\delta(e_i)$ for $i=m-d+1,\dots,m$;
   \item [(N3)] For each $e=(u,v)\in E\setminus S$ with $\mathcal{P}(e)=\{e_1,\dots,e_
   {p_e}\}$, there exists a function \mbox{$\phi_e:(\Sigma^n)^{p_e}\longrightarrow
       \Sigma^n$},  referred to as the \emph{local encoding function} of $e$, such that \mbox{$f_e(X)=\phi_{e}(f_{e_1}(X),\dots,f_{e_{p_e}}(X))$},
       where $p_e$ is the in-degree of $e$ and $\mathcal{P}(e)$ is the set
       of parent edges of $e$.
 \end{enumerate}
\end{definition}


If $n=1$, the network code is called a scalar network code,
otherwise (if $n>1$) it is called a vector or a block network code.
If  $\Sigma$ is a certain finite field $\mathbb{F}$, and all the global and local encoding functions are
linear functions of the packets, the network code is called linear over
$\mathbb{F}$.

\subsection{Index Coding}
An instance of the Index Coding problem $\mathcal{I}(X,R)$ includes
\begin{enumerate}
  \item A set of messages  $X=\{x_1,\dots,x_k\}$;
  \item A set of clients $R \subseteq\{(x,H); x\in X, H\subseteq
      X\setminus\{x\} \}$.
\end{enumerate}
Here, $X$ represents the set of messages available at the server. A
client is represented by a pair $(x,H)\in R$, where $x$ is the
message required by the client, and $H\subseteq X$ is set of
messages available to the client as side information.
We assume, without loss of generality, that each client needs exactly one message.

As in the Network Coding problem, each message $x_i\in X$ is divided into  $n$
packets $x_i=(x_{i1},\dots,x_{in})$. We refer to parameter $n$ as the block length of
the index code.

\begin{definition}[Index Code]\label{def:indexcoding}
Let $\mathcal{I}(X,R)$ be an instance of the Index Coding problem with $k=|X|$
messages, each message $x_i$ is a vector of $n$ packets, $(x_{i1},\dots,x_{in})
\in \Sigma^n$, where $\Sigma=\{0,\dots,q-1\}$ is a $q$-ary alphabet. Then, an
optimal $(n,q)$ index code for $\mathcal{I}(X,R)$ is a function
$f:(\Sigma^n)^k\longrightarrow \Sigma^\ell$, such that
\begin{itemize}
\item [(I1)]for each client $r=(x,H)\in R$, there exists a function
    $\psi_{r}: \Sigma^{\ell+n|H|}\longrightarrow \Sigma^n$ such that
    \mbox{$\psi_{r}(f(x_1,\dots,x_k),x_i|_{i\in H})=x$},
\item [(I2)] $\ell=\ell(n,q)$ is the smallest integer such that (I1) holds
    for the given  $q$-ary alphabet and block length $n$.
\end{itemize}
\end{definition}

We refer to  $\psi_r$ as the decoding function for client $r$. With a linear
index code, the alphabet $\Sigma$ is a field and the functions $f$ and $\psi_r$
are linear in  variables $x_{ij}$. Similarly, if $n=1$ the index code is
called a scalar code and for $n>1$ it is called a vector or block code.

Our formulation of the Index Coding problem here differs from that of
\cite{BBJK06} and \cite{LS07} in two aspects. First, the model of \cite{BBJK06}
and \cite{LS07} assumes that for each message in $X$ there is exactly one
client that requests it. Our model does not make this assumption. Second, and more importantly,  \cite{BBJK06} and \cite{LS07} focus on scalar linear codes (vector linear codes are mentioned in the conclusion of \cite{LS07}), whereas we consider the more general case of vector linear codes.

Let $\mathcal{I}(X,R)$ be an instance of the Index Coding problem.
We define by $\lambda(n,q)= \ell (n,q)/n$ the transmission rate of
the optimal solution over an alphabet of size $q$. We also denote by
$\lambda^*(n,q)$ the minimum rate achieved by a vector linear solution
over a finite field $\mathbb{F}_q$. We are interested in the
behavior of $\lambda$ and $\lambda^*$ as  functions of $n$ and $q$.

Let $\mu(\mathcal{I})$ be the largest set of messages requested by a collection
of clients with identical ``has'' sets, i.e.,
$\mu(\mathcal{I})=\max_{I\subseteq X}|\{x_i;\ (x_i,I)\in R\}|$. It is easy to
verify that the optimal rate $\lambda(n,q)$ is lower bounded by
$\mu(\mathcal{I})$, independently of the values of $n$ and $q$.

\begin{lemma}\label{lem:bound}
For any instance $\mathcal{I}(X,R)$ of the Index Coding problem it
holds that
$$\lambda(n,q)\geq \mu(\mathcal{I}).$$
\end{lemma}

\section{Main Result}\label{sec:connec}


In this section we present a reduction from the Network Coding problem to the
Index Coding problem. Specifically, for each instance
$\mathbb{N}(G(V,E),X,\delta)$ of the Network Coding problem, we construct a
corresponding instance $\mathcal{I}_{\mathbb{N}}$ of the Index Coding problem
such that $\mathcal{I}_{\mathbb{N}}$ has an $(n,q)$ index code of rate
$\lambda^*(n,q)=\lambda(n,q)=\mu(\mathcal{I}_{\mathbb{N}})$ if and only if
there exists an $(n,q)$ linear network for $\mathbb{N}$.

\begin{definition}\label{def:1}
Let $\mathbb{N}(G(V,E),X,\delta)$  be an instance of the Network Coding
problem. We construct an instance $\mathcal{I}_{\mathbb{N}}(Y,R)$ of the Index
Coding problem as follows:
\begin{enumerate}
  \item The set of messages $Y$ includes a message for each edge $e\in E$
      and the messages $x_i\in X$, i.e., \mbox{$Y=\{y_1,\dots,y_m\}\cup
      X$};
  \item The set of clients $R$ is a union of $R_1,\dots,R_5$ defined as
      follows:
  \begin{enumerate}
  \item  $R_1=\{ (x_i,\{ y_i\}); e_i\in S\}$
  \item  $R_2=\{ (y_i,\{ x_i\});e_i\in S\}$
  \item $R_3=\{(y_i, \{ y_j; e_j\in\mathcal{P}(e_i)\}); e_i\in
      E\setminus S\}$
  \item $R_4=\{(\delta(e_i),\{y_i\}); e_i\in D\}$
  \item $R_5=\{ (y_i, X); i=1,\dots,m \}$
  \end{enumerate}

\end{enumerate}
\end{definition}

It is easy to verify that for instance $\mathcal{I}_{\mathbb{N}}(Y,R)$ it holds
that $\mu(\mathcal{I}_{\mathbb{N}})=m$.

\begin{theorem}\label{th:equivalence}
Let $\mathbb{N}(G(V,E),X,\delta)$ be an instance of the Network Coding problem
and let $\mathcal{I}_{\mathbb{N}}(Y,R)$ be the corresponding instance of the Index Coding problem as
defined above. Then, there exists a linear $(n,q)$ index code  for
$\mathcal{I}_{\mathbb{N}}$ with $\lambda^*(n,q)=\mu(\mathcal{I}_{\mathbb{N}})$ iff there exists a linear
$(n,q)$ network code for $\mathbb{N}$.
\end{theorem}
\begin{proof}
Suppose there is a linear $(n,q)$ network code $ C=\{f_e(X);
f_e:(\mathbb{F}_q^n)^k\rightarrow \mathbb{F}_q^n, e\in E\}$ for $\mathbb{N}$ over the finite field $\mathbb{F}_q$ for some integer $n$.

Define $g:(\mathbb{F}_q^n)^{m+k}\rightarrow (\mathbb{F}_q^n)^m$ such that $\forall
Z=(x_1,\dots,x_k,y_1,\dots,y_m)\in (\mathbb{F}_q^n)^m, g(Z)=(g_1(Z),\dots,g_m(Z))$
with

\begin{align*}
 g_i(Z)&=y_i+x_i\quad  &i=1,\dots,k,\\
g_i(Z)&=y_i+f_{e_i}(X) \quad &i=k+1,\dots,m,
\end{align*}

Next, we show that $g(Z)$ is in fact an index code for
$\mathcal{I}_{\mathbb{N}}$ by proving the existence of the decoding functions
$\psi_r$. We consider five cases:
\begin{enumerate}
  \item $\forall r=(x_i,\{y_i\})\in R_1, \psi_r=g_i(Z)-y_i$,
  \item $\forall r=(y_i,\{x_i\})\in R_2, \psi_r=g_i(Z)-x_i$,
  \item $\forall r=(y_i, \{y_{i_1},\dots, y_{i_p}\})\in R_3$, since $C$ is
      a linear network code for $\mathbb{N}$, there exists a linear function $\phi_{e_i}$
      such that $f_{e_i}(X)=\phi_{e_i}(f_{e_{i_1}}(X),\dots,
      f_{e_{i_p}}(X))$. Thus,
      \mbox{$\psi_r=g_i(Z)-\phi_{e_i}(g_{{i_1}}(Z)-y_{i_1},\dots,g_{{i_p}}(Z)-y_{i_p})$},
   \item $\forall r=(\delta(e_i),\{y_i\})\in R_4, e_i\in D,
       \psi_r=g_i(Z)-y_i$,
   \item $\forall r=(y_i,X)\in R_5,\psi_r=g_i(Z)-f_{e_{i}}(X)$.
\end{enumerate}

Now assume that $g: (\mathbb{F}_q^n)^{m+k} \longrightarrow (\mathbb{F}_q^n)^m$ is a linear
$(n,q)$ index code for $I_{\mathbb{N}}$ over the field $\mathbb{F}_q$, such that
$\forall Z=(x_1,\dots,x_k,y_1,\dots,y_m)\in (\mathbb{F}_q^n)^{m+k},$
 $g(Z)=(g_1(Z),\dots,g_m(Z)), $ $x_i,y_i ,g_i(Z)\in \mathbb{F}_q^n.$ We
 write
 $$g_i(Z)=\sum_{j=1}^{k}x_jA_{ij}+\sum_{j=1}^m y_jB_{ij},$$
where  $i=1,\dots, m$ and $A_{ij}, B_{ij}\in M_{\mathbb{F}_q}(n,n)$ are sets of $n\times n$ matrices with elements in $\mathbb{F}_q$.

The functions $\psi_r$  exist for all $r\in R_5$ iff the matrix
$M=[B_{ij}]\in  M_{\mathbb{F}_q}(nm,nm)$, which has  the matrix $B_{ij}$ as a
block submatrix in the $(i,j)$th position,  is invertible. Define $h:
(\mathbb{F}_q^n)^{m+k} \longrightarrow (\mathbb{F}_q^n)^m$, such that $h(Z)=g(Z)M^{-1},
\forall Z\in (\mathbb{F}_q^n)^{m+k} $. So, we obtain
$$h_i(Z)=y_i+\sum_{j=1}^kx_j C_{ij}, i=1,\dots, m,$$
where $[C_{ij}]\in  M_{\mathbb{F}_q}(n,n)$.
We note this $h(Z)$ is a valid index code for $\mathcal{I}_{\mathbb{N}}$.  In
fact, $\forall r=(x,H)\in R$ with $\psi_r(g,x|_{x\in H})=x$,
$\psi_r^\prime=(h,x|_{x\in H})=\psi_r(hM,x|_{x\in H}))$ is a valid decoding
function corresponding to the client $r$ and the index code $h(Z)$.

For all $r\in\ R_1\cup R_4$, $ \psi_r^\prime$ exists  iff for
$i=1,\dots,k,m-d+1,\dots, m$ and $j\neq i$ it holds that  $C_{ij}=[0]\in
M_{\mathbb{F}_q}(n,n)$ and $C_{ii}$ is invertible, where $[0]$ denotes the all
zeros matrix. This implies that
\begin{equation}\label{eq:indexcod}
\begin{split}
 h_i(Z)&=y_i+ x_iC_{ii}, i=1,\dots, k\\
    h_i(Z)&=y_i+\sum_{j=1}^k x_j C_{ij}, i=k+1,\dots, m-d\\
    h_i(Z)&=y_i+ \delta(e_i)C_{ii}, i=m-d+1,\dots,m
\end{split}
\end{equation}
Next, we define the functions $f_{e_i}:(\mathbb{F}_q^n)^k \longrightarrow \mathbb{F}_q^n,
e_i\in E$ as follows:
\begin{enumerate}
\item $f_{e_i}(X)=x_i$, for $ i=1,\dots,k$
\item $f_{e_i}(X)=\sum_{j=1}^k x_jC_{ij}$, for $ i=k+1,\dots,m-d$
\item $f_{e_i}(X)=\delta(e_i)$, for $ i=m-d+1,\dots, m$.
\end{enumerate}

Then $C=\{f_{e_i}; e_i\in E\}$ is a linear $(n,q)$ network code for
$\mathbb{N}$. To show that it suffices to prove that condition N3 holds.

Let $e_i$ be an edge in $E\setminus S$ with the set
$\mathcal{P}(e_i)=\{e_{i_1},\dots,e_{i_p}\}$ of parent edges. We denote by
$I_i=\{i_1,\dots,i_p\}$
and $r_i=(y_i,\{y_{i_1},\dots,y_{i_p}\})\in R_3$. Then, there is a linear function
$\psi_{r_i}^\prime$ such that $y_i= \psi_{r_i}^\prime
(h_1,\dots,h_m,y_{i_1},\dots, y_{i_p})$. Hence, there exist $T_{ij},
T^\prime_{i\alpha}\in M_{\mathbb{F}_q}(n,n)$ such that
$$y_i=\sum_{j=1}^mh_j T_{ij}+\sum_{\alpha\in I_i}y_\alpha T^\prime_{i\alpha}$$
Using Eq.~\eqref{eq:indexcod}, we get that $T_{ii}$ is the  identity
matrix, $T^\prime_{i\alpha}=-T_{i\alpha}\forall \alpha\in I_i$,
  $T_{ij}=[0]\ \forall j\notin I_i\cup\{i\}$.
  Therefore, $$f_{e_i}=-\sum_{\alpha\in I_i} f_{e_\alpha  }T_{i\alpha}, \forall e_i\in E\setminus S,$$ and $C$ is a feasible network code for $\mathbb{N}$.
\end{proof}

\begin{lemma}\label{ref:nonlinear}
Let $\mathbb{N}(G(V,E),X,\delta)$ be an instance of the Network Coding problem
and let $\mathcal{I}_{\mathbb{N}}(Y,R)$ be the corresponding index problem. If there is an
$(n,q)$ network code (not necessarily linear)
 for $\mathbb{N}$, then there is a $(n,q)$ index code
for $\mathcal{I}_{\mathbb{N}}$ with
\mbox{$\lambda(n,q)=\mu(\mathcal{I}_{\mathbb{N}})=m$}, where $m=|E|$.
\end{lemma}
\begin{proof}
The proof can be obtained by slightly modifying the first part of the proof of Theorem
\ref{th:equivalence}.
\end{proof}

\section{Applications}
\subsection{Block Encoding}

Index coding, as noted in \cite{BBJK06, LS07}, is reminiscent of the zero-error
source coding with side information problem discussed by Witsenhausen  in
\cite{W76}. Two cases were studied there depending on whether the transmitter knows the side information available to the receiver or not. It was  shown  that in the former case repeated scalar encoding is optimal, i.e.\ block encoding does not provide any benefit. We will demonstrate in this section that this result does not always hold for the Index Coding problem which can be regarded as an extension of the point to point problem discussed in \cite{W76}.

Let $\mathcal{N}_1$ be the M-network introduced in \cite{EMHK03} and depicted
in Figure \ref{fig:MPappus}(a). It was shown in \cite{DFZ07} that this network
does not have a scalar linear solution, but has a vector linear
one of block length $2$. Interestingly, such a vector linear solution does
not require encoding. In fact, reference \cite{DFZ07} proves a more general theorem:
\begin{theorem}\label{th:Mnet}
The M-network has a linear network code of block length $n$ iff $n$ is even.
\end{theorem}

%
Next, we present another network $\mathcal{N}_2$, that we refer to as  the
\emph{non-Pappus network}, and that has the same property as the M-network. Both
of these networks will be used to construct two instances of the Index Coding
problem where vector linear  outperforms scalar linear coding.

\begin{definition}[non-Pappus Network]
Let $S_0=\{\{1,2,3\},\{1,5,7\},\{3,5,9\},\{2,4,7\} ,\{4,5,6\},\{2,6,9\},$
$\{1,6,8\},\{3,4,8\} \}$, and $S_1=\{I\subseteq \{1,2,\dots,9\};
|I|=3\}\setminus S_0$. The non-Pappus network $\mathcal{N}_2$ is obtained by
adding to the network depicted in Figure \ref{fig:MPappus}(b) a node $n_I$ for
each $I=\{i,j,k\} \in S_1$, the edges $(n_i, n_I), (n_j, n_I), (n_k, n_I)$ and
three output edges outgoing from $n_I$, each one of them demands a different $x_i$.
\end{definition}

\begin{figure}[h]
\begin{center}
\includegraphics[scale=.6]{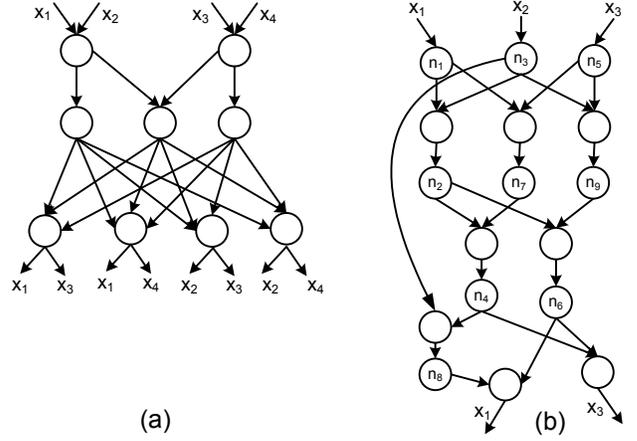}
 \caption{(a) The M-Netwrok $\mathcal{N}_1$. (b) A subnetwork of the non-Pappus
network $\mathcal{N}_2$.} \label{fig:MPappus}
\end{center}
\end{figure}

\begin{theorem}\label{th:nonPappus}
There is no  scalar linear network code for the non-Pappus network over any
field, but there is a $(2,3)$ linear one.
\end{theorem}

\begin{proof}
Let $C=\{f_e; e\in \mathcal{N}_2 \}$ be a scalar linear network code for
$\mathcal{N}_2$ over a certain field $\mathbb{F}$. Without loss of generality,
we assume that for each node $n_i$ of $\mathcal{N}_2$, the functions associated
with its output edges are identical. We define then $f_i=f_e$ where $e$ is an
outgoing edge to $n_i$, $i=1,\dots,9$, and write
$f_i=a_{i1}x_1+a_{i2}x_2+a_{i3}x_3=a_i\cdot X^T$, where $X=(x_1,x_2,x_3)$ and
$a_i=(a_{i1},a_{i2},a_{i3})$.

Since $\forall I=\{i,j,k\}\in S_1$, the outgoing edges to node $n_I$ demand
$x_1$, $x_2$ and $x_3$, we have $\rank \{a_i,a_j,a_k\}=3$. Furthermore, from
the connectivity of $\mathcal{N}_2$, we deduce that $a_2$ should be a linear
combination of $a_1$ and $a_3$, giving $\rank\{a_1,a_2,a_3\}< 3$. But
$\rank\{a_1,a_2,a_4\}=3$, which implies that $\rank\{a_1,a_2,a_3\}>1$, hence
$\rank\{a_1,a_2,a_3\}=2$. Similarly, $\forall \{i,j,k\} \in S_0,
\rank\{a_i,a_j,a_k\}=2$.

Therefore, letting $A=\{a_1,a_2,\dots,a_9\}$, the matroid
$\mathcal{M}(A,\rank)$ is the non-Pappus matroid shown in Figure
\ref{fig:nonPappusMat} \cite[p.43]{OX93}. Therefore, the vectors $a_i$ form a
linear representation  of $\mathcal{M}$ over $\mathbb{F}$. But, by Pappus
theorem \cite[p.173]{OX93}, the non-Pappus matroid is not linearly
representable over any field, which leads to a contradiction. So,
$\mathcal{N}_2$ does not have  a scalar linear solution.

Let $x_1=(x,y), x_2=(w,z), x_3=(u,v) \in \mathbb{F}_3^2$. Define $f_1(X)=x_1,
f_2(X)= (x+w,  y+z), f_3(X)= x_2, f_4(X)=(x+u+2z, y+2v+w+z),
f_5(X)=x_3,f_6(X)=(x+2u+2v+2z, y+u+w+z), f_7(X)=(x+v,y+u+2v), f_8(X)=( x+u+w+z,
y +2v+w), f_9(X)=(u+w, v+z)$. These functions correspond to the multilinear (or
partition) representation of the non-Pappus matroid discussed in \cite{SA98,
M99}. For each edge $e\in G$ outgoing from node $n_i, i=1,\dots,9$, define
$f_e=f_i$. And for each edge $e\in D$, let $f_e=\delta(e)$. Then, $\{f_e;e\in
\mathcal{N}_2\}$ is a $(2,3)$ network code for the non-Pappus network.
\end{proof}

\begin{figure}[t]
\begin{center}
\includegraphics[scale=.52]{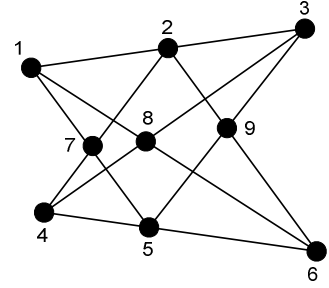}
 \caption{ A graphical representation of the non-Pappus matroid  of rank 3 \cite[p.43]{OX93}. Cycles are
 represented by straight lines.} \label{fig:nonPappusMat}
\end{center}
\end{figure}

 Now, consider  $I_{\mathcal{N}_1}$ and $I_{\mathcal{N}_2}$ the two Index Coding problems  corresponding respectively to the M-network and the non-Pappus network obtained by the construction of the
previous section. Both do not admit  scalar linear index codes  that achieve
the bound of Lemma \ref{lem:bound}, but have linear index codes of length 2,
$I_{\mathcal{N}_1}$ over $\mathbb{F}_2$ and $I_{\mathcal{N}_2}$ over
$\mathbb{F}_3$, that meet this bound. Thus, $I_{\mathcal{N}_1}$ and
$I_{\mathcal{N}_2}$ are two instances of the Index Coding problem where vector
linear coding outperforms scalar linear coding. This result can be summarized
by the following corollary:
\begin{corollary}
For $I_{\mathcal{N}_1}, \lambda^*(2,2)<\lambda^*(1,2)$. And for $I_{\mathcal{N}_2},  \lambda^*(2,3)<\lambda^*(1,3)$.
\end{corollary}
\begin{proof} Follows directly from Theorems \ref{th:equivalence},
\ref{th:Mnet} and \ref{th:nonPappus}.
\end{proof}

\subsection{Linearity vs. Non-Linearity}
Linearity is a desired property for any code, including index codes.  It was
conjectured in \cite{BBJK06} that binary scalar linear index codes are optimal,
meaning that $\lambda^*(1,2)=\lambda(1,2)$ for any instance of the Index Coding problem. The authors of \cite{LS07} disproved this conjecture for \emph{scalar
linear} codes by providing for any given number of messages  $k$ and field
$\mathbb{F}_q$, an instance of the Index Coding problem with a large gap
between $\lambda^*(1,q)$ and $\lambda(1,q)$.

In this section we show that \emph{vector linear} codes are insufficient for
minimizing the number of transmissions. In particular, we prove that non-linear
index codes outperform vector linear codes for any choice of field and block
length $n$. Our proof is based on the insufficiency of linear network codes
shown in \cite{DFZ05}.

First, we present the network $\mathcal{N}_3$ depicted in Figure
\ref{fig:zegernet} which was introduced and studied in \cite{DFZ05}. The
following theorem was proved in \cite{DFZ05}.

\begin{theorem}\label{theor:aux}
The network $\mathcal{N}_3$ does not have a linear solution, but has a $(2,4)$
non-linear solution.
\end{theorem}

Let $\mathcal{I}_{\mathcal{N}_3}$  be an instance of the Index Coding problem
that corresponds to $\mathcal{N}_3$ constructed according to
Definition~\ref{def:1}. Theorem \ref{theor:aux} implies that
$\mathcal{I}_{\mathcal{N}_3}$  does not have a linear solution that achieves
$\mu(\mathcal{I}_{\mathcal{N}_3})$, the lower bound of Lemma \ref{lem:bound}.
However, by Lemma \ref{ref:nonlinear}, the $(2,4)$ non-linear code of
$\mathcal{N}_3$ can be used to construct a $(2,4)$ non-linear index code for
$\mathcal{I}_{\mathcal{N}_3}$ that achieves the lower bound of lemma
\ref{lem:bound}. Hence, we obtain the following result:

\begin{corollary}
For the instance $\mathcal{I}_{\mathcal{N}_3}$ of the  Index Coding problem it
holds that $\lambda(2,4)=\mu(\mathcal{I}_{\mathcal{N}_3})$. Furthermore, for
any positive integers $n$ and  $q$, it holds that $\lambda^*(n,q)<
\lambda(n,q)$.
\end{corollary}

\begin{figure}[t]
\begin{center}
\includegraphics[scale=.5]{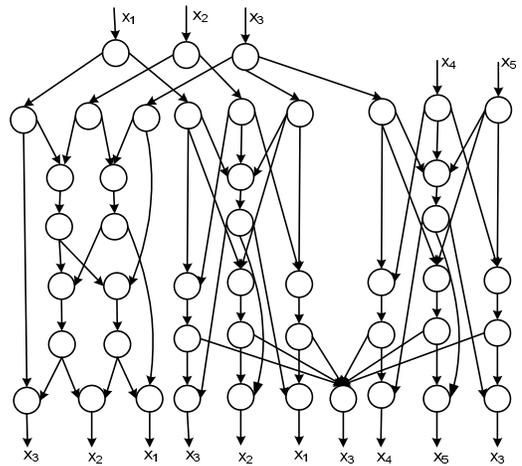}
 \caption{The network $\mathcal{N}_3$ of \cite{DFZ05}. $\mathcal{N}_3$ does not have
 a linear network code over any field, but has a non-linear one over a quaternary alphabet. } \label{fig:zegernet}

\end{center}
\end{figure}

\section{Conclusion}
In this paper we studied the connection between the Index Coding and Network
Coding problems. We showed a reduction that maps each communication network
$\mathbb{N}$ to an instance of the Index Coding problem
$\mathcal{I}_\mathbb{N}$ such that $\mathbb{N}$ has a linear network code if
and only if $\mathcal{I}_\mathbb{N}$ has a linear index code over the same
field that satisfies a certain condition on the number of transmissions.

This reduction allowed us to apply many important results for network coding to index coding. For instance, we introduced the non-Pappus network and showed
that it does not have a scalar linear network code, but has a vector linear
one. The non-Pappus network in addition to the M-network of \cite{EMHK03} were
used to construct index coding instances where vector linear solutions
outperform scalar linear solutions. Another application of this reduction
concerns the comparison of linear and non-linear index codes. Using the results
of Dougherty et al. in \cite{DFZ05} we proved the insufficiency of vector
linear solutions for the Index Coding problem.

\bibliographystyle{unsrt}

\bibliography{coding}

\end{document}